\documentclass[aps,pra, nopacs,twocolumn,twoside,floatfix,a4,superscriptaddress]{revtex4}
\usepackage{amsmath, amssymb,amsthm}
\usepackage{color}
\usepackage{graphicx,epsfig}
\usepackage{times} 
\theoremstyle{plain}
\newtheorem{theorem}{Theorem}
\newtheorem{cor}{Corollary}
\theoremstyle{definition}
\newtheorem{definition}{Definition}
\newtheorem{hypothesis}{Hypothesis}
\newtheorem{lemma}{Lemma}

\input{Qcircuit}

\newcommand{\nt}[1]{^{\otimes #1}}

\newcommand{\IQPstar}{IQP$^{*}$}

\usepackage[normalem]{ulem}
\usepackage{color}

\newcommand{\dan}[1]{{#1}}

\begin{document}

\title{Measurement-based classical computation} 

\author{Matty J. Hoban}
%\email{matthew.hoban@icfo.es}
\affiliation{ICFO-Institut de Ci\`{e}ncies Fot\`{o}niques, Mediterranean Technology Park, E-08860 Castelldefels (Barcelona), Spain}
\author{Joel J. Wallman}
\affiliation{Centre for Engineered Quantum Systems, School of Physics, 
The University of Sydney, Sydney, NSW 2006, Australia.}
\author{Hussain Anwar}
\author{Na\"{i}ri Usher}
\affiliation{Department of Physics and Astronomy, University College London, Gower Street, London WC1E 6BT, United Kingdom.}
\author{Robert Raussendorf}
\affiliation{Department of Physics and Astronomy, University of British Columbia, Vancouver, BC V6T 1Z1, Canada}
\author{Dan E. Browne}
\affiliation{Department of Physics and Astronomy, University College London, Gower Street, London WC1E 6BT, United Kingdom.}

\date{\today}

\begin{abstract}
Measurement-based quantum computation (MBQC) is a model of quantum computation, in which computation proceeds via adaptive single qubit measurements on a multi-qubit quantum state. It is computationally equivalent to the circuit model. Unlike the circuit model, however, its classical analog is little studied. Here we present a classical analog of MBQC whose computational complexity presents a rich structure. To do so, we identify uniform families of quantum computations (refining the circuits introduced by Bremner, Jozsa and Shepherd in \textit{Proc. R. Soc. A} \textbf{467}, 459 (2011)) whose output is likely hard to exactly simulate (sample) classically.
We demonstrate that these circuit families can be efficiently implemented in the MBQC model without adaptive measurement, and thus can be achieved in a classical analog of MBQC whose resource state is a probability distribution which has been created quantum mechanically. 
Such states (by definition) violate no Bell inequality, but, if widely held beliefs about computational complexity are true, they nevertheless exhibit non-classicality when used as a computational resource---an imprint of their quantum origin. 
\end{abstract}

\maketitle

There is a strong belief that quantum computers can efficiently perform certain tasks that cannot be performed efficiently on a classical computer, such as integer factorization~\cite{Shor}. One of the central questions of quantum information theory is to better understand which aspects of quantum evolution are efficiently classically simulatable and which are not \cite{vandennest1}. One important aspect of this investigation has been to learn when a computational model cannot possess a super-classical speed-up, by showing that it can be simulated efficiently on a classical computer. For example, Jozsa and Linden showed that in pure state circuit-model quantum computation, restricting the multi-partite entanglement in certain ways renders the model classically efficiently simulatable \cite{jozsa}. In contrast, a number of striking recent results \cite{terhaldivincenzo,IQP2,aa} have given rigorous evidence that certain models of quantum computation (that have circuits with unrestricted entanglement) are unlikely to admit an efficient classical simulation.

%
%Mention entanglement in MBQC
A distinct way to question the role of entanglement in quantum computing is to consider it within the model of Measurement-based Quantum Computation (MBQC) \cite{mbqc}. In MBQC, computation proceeds via a sequence of single-site measurements on a (usually entangled) many qubit resource state. Certain entangled resource states, such as the cluster state \cite{briegel1}, are known as universal resources since they enable universal quantum computation in this model.  It has been shown that the computational properties of a resource state  can be linked to its entanglement properties \cite{mbqcresources}.
%
%
%The computational power derives from interplay of the adaptive measurements, and the entanglement properties of the initial resource state (e.g., entanglement). A common approach is to specify the amount or type of entanglement in the resource state and then determine what computations can be performed. For example, quantum states can be too entangled to be useful \cite{tooentangled}, but an unbounded amount of entanglement is needed for universal quantum computing \cite{vandennestugc}.
 Here, we consider the computations that can be performed in the MBQC framework when \textit{no entanglement} is present in the resource state by developing and studying a classical analog of MBQC.

It is important to define clearly what we mean by \textit{non-classical} in the context of computation. In this paper, we denote a standard classical computing device (for a formal definition see Appendix \ref{appsec0}) as a classical computer that has access to  uniformly random bits (i.e. as is used to define the complexity class BPP). We then define non-classical computation as any family of computations which cannot be achieved efficiently (i.e. in polynomial time) with such a device.

The connections between MBQC and classical computation was studied from one perspective in \cite{andersbrowne}, where it was shown that casting classical computations within the MBQC model illuminated a close connection between MBQC and GHZ-type paradoxes (see also \cite{raussenGHZ}). Here we take an alternative approach. MBQC can be split into three components: a multi-qubit resource state; adaptive local measurements; and the classical side-computation which processes input and output and allows adaptive measurement \cite{RBB}. In a full quantum realisation of MBQC, the first two components are quantum, and the latter classical. In this paper we consider the consequences of making all three components classical. 

What is the classical analog of an entangled resource state? When we measure a quantum state, the output is usually random. Moreover, we can only make a measurement once---in entangled-state MBQC, measurement always changes the state. Due to the single-use property of the entangled resource states MBQC is often called the ``one-way quantum computer" \cite{mbqc}. The classical object which shares these properties is a \textit{single sample} from a multi-bit probability distribution. Like a set of single qubit measurements on an entangled state, it returns a random bit-string, similarly it supplies this only once. There are significant fundamental differences between a classical sample and an entangled state, however, both can be considered as resources in an MBQC-like framework. In this paper we  define \textit{measurement-based classical computation} (MBCC) as a model of computation consisting of polling a single sample from a multi-bit probability distribution and performing classical post-processing on these bits. Furthermore, we restrict classical post-processing to the sub-class of computations utilised in cluster state MBQC \cite{RBB}, linear computations (generated by XOR and NOT-gates alone).

Two of us showed in \cite{matty} that, under this restriction, if an MBQC resource violates no Bell inequality then no non-linear computation can be achieved. The computation is restricted to convex combinations of linear functions of the input bits. The expressiveness of MBCC (the types of computations it can perform), in which, as a classical model, no Bell inequality can be violated,  is therefore limited in the same way. Can we then prove that MBCC can be simulated efficiently by a classical computer?

\dan{The resource in MBCC is a classical multi-bit probability distribution. Such distributions are of exponential size and include distributions unlikely to by efficiently realisable even by quantum resources. We thus say that an $n$-bit distribution is \textit{efficiently quantum preparable} when there exists a quantum circuit with a  polynomial in $n$ description, upon which the output of single qubit measurements can can prepare the distribution exactly.}

In this paper, we give strong evidence that MBCC with an efficiently quantumly preparable resource can be computationally non-classical. More precisely we show that:
\begin{theorem}\label{newtheorem1} There exist uniform families of MBCC computations \dan{with efficiently quantum preparable resources} which cannot be efficiently exactly simulated via a standard classical computing device unless the polynomial hierarchy collapses to the third level.   \end{theorem}

By standard classical computing device we mean a classical Turing machine whose sole random element is a supply of uniformly random bits.  This is a probabilistic Turing machine, and is used to define complexity classes BPP and PP. MBCC is also a fully classical computational model, \dan{but crucially the multi-bit probability distribution may have an (efficient) non-classical preparation.}

The Polynomial Hierarchy is a family of classes in computational complexity theory  \cite{papabook}. It is believed, although not proven, that this family of classes is distinct. Aaronson and Arkhipov (AA) recently called this a ``generic, foundational'' assumption of computer science  \cite{aa}, and this has been used to provide strong evidence that universal quantum computers \cite{terhaldivincenzo}, and certain restricted sub-classes of quantum computers \cite{IQP2} and quantum processes  \cite{aa} are hard to exactly simulate on a classical computer. 

\begin{hypothesis}  The third level of the polynomial hierarchy is strictly smaller than at least one other level in the hierarchy.
\end{hypothesis}
     Under the assumption of Hypothesis 1, Bremner, Jozsa and Shepherd (BJS) showed that uniform families of a very restricted family of quantum circuits, Instantaneous Quantum Polytime (IQP) circuits, could not be exactly efficiently simulated on a standard classical computer, where by simulate we mean that the classical devices outputs a sample from an identical distribution to the simulated quantum circuit (\textit{weak simulation} in Jozsa and Van Den Nest's classification \cite{jozsavdn}).

Our technical results include a strengthening of BJS's result by introducing a new and much stricter uniformity condition defining uniform families of circuits which we call (IQP$^*$). We show that \IQPstar circuits also cannot be classically simulated unless Hypothesis 1 is violated.
We then show that \IQPstar circuits can be implemented in cluster-state MBQC using fixed-basis measurements. We thus demonstrate that the same computations can be implemented in MBCC, the classical analog of MBQC introduced above.

%Our main result then appears paradoxical, that MBCC includes computations impossible to implement efficiently on a randomised classical computer. The resolution is that MBCC samples from a probability distribution which can be 

% The fixed bases of the measurements allow us then to simply derive families of separable and discord-free 
%resource states for the same computations. \IQPstar circuits can thus be implemented in MBCC (using a quantum-generated probability distribution as a resource) proving Theorem 1.

%The simplicity of BJS's models (and our modification) as well as AA's model raise the hope that a quantum experiment might soon demonstrate a computation beyond the power of conventional computers. Early experimental demonstrations of boson sampling \cite{bosonsamplingexp} have been in agreement with quantum mechanical predictions.

We begin by defining IQP (Instantaneous Quantum Polytime) circuits \cite{IQP1}, introduced by BJS, which will play a central role in our argument. 
\begin{definition}
An \textit{IQP circuit} with classical input bit string $x$ of size $n$ acting on $q\geq n$ qubits consists of:
\begin{enumerate}
\item a quantum register prepared in the input state $\ket{x}\ket{0}\nt{q-n}$; and
\item the application of a unitary operator $U$ to the register, where $U$ is diagonal with respect to the eigenbasis of Pauli-$X$ operators.
%\item The measurement of each qubit in the computational basis.
\end{enumerate}
We denote the output of this computation, obtained via computational basis measurements on every qubit,
 by the $q$-bit string
$m$, whose $j$th element $m_j\in\{0,1\}$ is the outcome of the computational basis measurement on the 
$j$th qubit. 
\end{definition}

An example of an IQP circuit is illustrated in Fig.~\ref{fig:egIQP}a. Such circuits 
are called \emph{instantaneous} because $D$ can be decomposed into a product of commuting gates, which can thus be applied in any order (or simultaneously)~\cite{IQP1}.

\begin{figure}
\centering
	\includegraphics[width=0.5\textwidth]{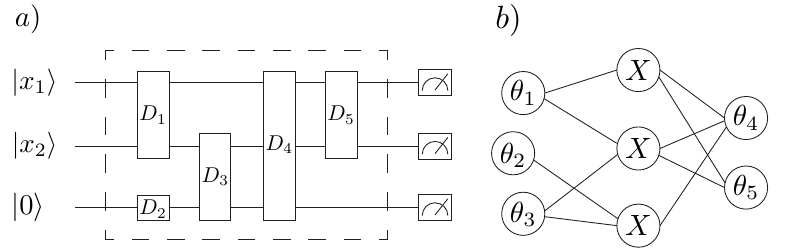}
	\caption{a) Standard form of an IQP circuit, where each gate is diagonal in the Pauli-X basis and $(x_1,x_2)$ 
	is the two bit input string. All measurements are in the computational (Pauli-Z) basis. The boxed gates give 
	the unitary $D$. b) A MBQC implementation of the circuit in a), where each circle 
	represents a qubit prepared in the state $\ket{+}$ and edges between circles represent the application of a 
	controlled-$Z$ gate. The contents of the circles represent the basis in which the corresponding qubit is 
	measured, where $X$ represents the Pauli-X basis, and $\theta_{j}$ represents the basis $U_{X}(-\theta_{j})ZU_{X}(\theta_{j})$,
	where $U_{X}(\theta_{j})$ is a rotation by $\theta_{j}$ about the Pauli-X axis. 
	The angles $\theta_j$ are in one-to-one correspondence with the $\theta_{z}$ in the representation of 
	$D_{j}$ in Eq. \eqref{Dgate}. All of these measurements can be implemented simultaneously (non-adaptively) in MBQC.}
	\label{fig:egIQP}
\end{figure}

When studying the computational power of families of circuits, it is often useful to ensure that unreasonable computational power is not hidden in the description of the circuits themselves. This can be ensured via a \emph{uniformity condition} which ensures that a description of each circuit in the family can be (classically) efficiently generated. While BJS introduce a uniformity condition in Ref. \cite{IQP2}, we adopt a different one here, and we denote the set of uniform circuit families obtained under this condition by IQP$^*$.

\begin{definition}
An {\IQPstar} circuit family is a family of IQP circuits, with input $x$ and input size $n=|x|$, followed by computational basis measurements on every qubit, 
such that the number of qubits $q$ is polynomial in $n$, and where the unitary operator $U_n$ (which has an explicit $n$-dependence) is a $\mbox{poly}(n)$ product of gates of the form
\begin{equation}\label{Dgate}
D(\theta_{z},z)=e^{i\theta_{z}X[z]}{,}
\end{equation}
where each angle $\theta_{z}\in(0,2\pi]$ has a description polynomial-size in $n$, $z$ is a $q$-bit string, and we introduce the notation $X[z]=\bigotimes_{j}X^{z_{j}}${,} where $z_{j}$ is the $j$th bit of $z$. E.g. $D(\theta_{{101}},101)=\exp[i\theta_{{101}} (X\otimes \openone\otimes X)]$. 
\end{definition}

In other words, the description of every member of an IQP$^*$ circuit family
is a polynomial list $L_{n}$ of $q$-bit strings $z$ and corresponding angles $\theta_{z}$. This list then defines the circuit for input size $n$. If we adopt the notational shorthand that $\theta_{z}=0$ for all bit strings $z$ not in $L_{n}$, then the unitary transformation for the circuit $U_n$ is given by
\begin{equation}\label{Ddef}
U_n=\prod_{z\in Z_{2}^{q}} D(\theta_{z},z).
\end{equation}
Operators of this form have a useful symmetry in their matrix elements, which we will exploit below, namely,
\begin{equation}\label{useful}
\bra{w\oplus y}D_n\ket{w}=\bra{y}D_n\ket{0}
\end{equation}
for all bit strings $w$ and $y$, with $\ket{0}{\equiv}\ket{0}^{\otimes q}$ and where $\oplus$ represents a bit-wise sum modulo 2.

This is different from the uniformity condition introduced by BJS \cite{IQP2} in some significant ways. Most importantly, BJS' uniformity condition allows the circuit to depend on individual values of input string $x$, rather than (the more common choice of) the length of $x$. This means that the circuit construction itself can play a very siginificant role in the computation, for example evaluating an arbitrary polynomial-sized classical circuit. In contrast, due to their much weaker ``pre-computation'' stage,  {\IQPstar} circuit families cannot even achieve a single non-linear logic Boolean function, such as AND. Lemma 1 below is thus a considerable strengthening of the theorems in \cite{IQP2}. %Note the linearity of equation \eqref{useful}.

\begin{lemma}\label{lemma1}
\textit{The output probability distributions generated by IQP$^*$ circuit families cannot
be efficiently and exactly simulated on a standard classical computing device unless Hypothesis 1 is false.}
\end{lemma}

The full proof of Lemma \ref{lemma1} is presented in Appendix \ref{appsec1}. The technical definition of a classical computing device is provided in Appendix \ref{appsec0} along with other useful notions from computational complexity theory. To summarize the proof for readers familiar with Ref. \cite{IQP2}, under postselection of measurement outcomes of a subset of qubits, the families of IQP$^{*}$ circuits can be mapped to general quantum circuits satisfying the standard uniformity condition for the complexity class BQP. We then utilize a similar proof technique to  Ref. \cite{IQP2}. Importantly, post-selection allows \IQPstar circuits to implement algorithms for hard problems (hard for complexity class PP) \cite{aaronson3}.

As in Ref. \cite{IQP2}, Lemma \ref{lemma1} may be generalized to include multiplicative error up to a certain factor on the individual probabilities. However, we shall not consider such multiplicative error here. We discuss the issue of approximate simulations and finite numbers of samples at the end of this paper. Lemma \ref{lemma1} presents strong evidence that IQP$^*$ circuit families may not be efficiently and exactly simulated on a classical computer. 

Before proceeding to our main result, we note a related phenomenon in Corollary~\ref{theorem2}, namely that there exist efficiently-preparable families of quantum states for whom the statistics of computational basis measurements are unlikely to be efficiently and exactly simulated on a classical computer.

\begin{definition}
An IQP$^*$ \textit{zero-input state family} is the set of quantum states created by an IQP$^*$ circuit family, when the input is set to the all zeros string $0\ldots0$.
\end{definition}

\begin{cor}\label{theorem2}
The statistics of computational basis measurements on \emph{IQP$^*$} zero-input 
state families cannot be efficiently and exactly simulated on a standard classical computer unless Hypothesis 1 is false.\end{cor}

\begin{proof}
We use a special property of IQP circuits, namely that the output statistics of an IQP circuit with input $x$, defined according to Definition 1, may be realized by the same IQP circuit acting on the $n$-bit all-zeros string by performing some simple extra post-processing of the output bits of the measurements.

Observe that the probability that the measurement output string is a bit string $m$, given input $x$, is
\begin{align}\label{mainequation}
\textrm{Prob.}(m|x)&=|\langle m|D_n|\bar{x}\rangle|^2,	\nonumber\\
&=|\langle m\oplus \bar{x}|D_n|0\rangle|^2	,\nonumber\\
&=	\textrm{Prob.}(m\oplus\bar{x}|0)	,
\end{align}
where $\bar{x}$ is $x$ appended by $q-n$ zeros and we obtained the second line from Eq. \eqref{useful}.

Thus identical output statistics to an IQP circuit given input $x$ can be obtained via the same circuit with input $0$ and post-processing of the output bit string $m$ to string $m\oplus \bar{x}$. This post-processing comprises at most $n$ bit-flips and can be (trivially) efficiently performed on a classical device. From this, the corollary follows directly from Lemma 1.
\end{proof}

Corollary 1 is an important preface to our main result, and captures many of its features. Note that the subject of the corollary is a classical probability distribution. Even though the distribution is \textit{{classical}}, its statistics have inherited the (likely) hardness of exact simulation of the \IQPstar \textit{quantum} circuit families.

We now turn to MBQC (and then MBCC) implementations of \IQPstar circuit families. We show the general result that any MBQC which can be achieved using measurements in a fixed basis can be achieved in MBCC (though possibly requiring a resource state that cannot be efficiently generated classically).  

IQP circuits have a special form in MBQC, first derived in \cite{gbriegel} and illustrated in Fig. 1. We will show below that \IQPstar families also have the property that they can be achieved in MBQC with a non-adaptive fixed measurement basis.
 
 \begin{lemma}\label{new2lemma}
For every instance of MBQC on an $n$-qubit resource, where every measurement basis is fixed, there is a corresponding instance of MBCC with an $n$-bit resource, whose output statistics simulate it exactly.\end{lemma}
 
 The proof of Lemma 2 follows immediately from the fact that the statistics of any projective measurement upon a quantum state relies solely on matrix elements which are diagonal with respect to the measured basis. Given the state, a fully dephasing channel in this eigenbasis may be applied, which sets all off-diagonal elements to zero. This channel will not change the statistics of the measurement which follows but outputs a state which is separable and discord-free \cite{cablermp}. Measurements on a separable and discord-free quantum state define a multi-bit probability distribution. MBCC utilising a sample of this distribution will exactly simulate an instance of the MBQC. 

The final step in the proof of Theorem 1 is to show that \IQPstar circuits can be implemented in MBQC with fixed measurements.

 \begin{lemma}\label{newnewlemma}
Given a member of an \IQPstar circuit family, there is an efficient implementation in MBQC whose measurements are fixed and non-adaptive.
\end{lemma}

\begin{proof}
From Eq. \eqref{mainequation}, we can place all of the dependence on $x$ into classical post-processing.
For convenience, we insert $H^2=\openone$ between every state, unitary and measurement in an \IQPstar circuit to change the input state from $\ket{0}\nt{q}$ to $\ket{+}\nt{q}$, the final measurements into Pauli-$X$ measurements and to make all gates diagonal in the Pauli-$Z$ basis (rather than the Pauli-$X$ basis). Then, noting that $CZ^2 = \openone$, where $CZ$ is the controlled-phase gate, we can add $CZ$ gates after each input state, so that the \IQPstar circuit is equivalent to preparing a cluster state, applying gates that are diagonal in the Pauli-$Z$ basis and then measuring in the Pauli-$X$ basis.

In \cite{gbriegel}, it is shown that any unitary gate of the form
\begin{equation}\label{DgateZ}
D_z(\theta_{z},z)=e^{i\theta_{z}Z[z]}{,}
\end{equation}
where $Z[z]=\bigotimes_{j}Z^{z_{j}}$ and $z_{j}$ is the $j$th bit of $z$, can be achieved in MBQC with simultaneous measurements, and with byproduct operators which are a tensor product of $Z$ and $\openone$. Since the byproduct operators commute with the logical gates \eqref{DgateZ}, they can be applied at the end of the computation and measurements do not need to be adaptive. Therefore all the gates in an \IQPstar circuit can be implemented in MBQC without adaptivity.
\end{proof}

Combining Lemma \ref{new2lemma} with Lemma \ref{newnewlemma} implies that \IQPstar circuit families can be realized in MBCC, with a resource distribution which is efficiently quantum preparable. Together with Lemma \ref{lemma1} this proves Theorem \ref{newtheorem1}.

\textit{Discussion.} 
What is the relationship between MBCC and MBQC? From one perspective, MBCC can be seen as a special case of MBQC, since dephased states are a sub-set of quantum states. In both models the classical side-processing is of secondary importance (and of weak computational power) and the driver of the computation is the correlations in the measurement outcomes. In other aspects the two models are strikingly different.  In response to an earlier pre-print of this present work, Rieffel and Wiseman \cite{howard} introduced criteria  which divides sets of MBQC instances into superficially and inherently measurement based. Using their criteria, all instances of MBCC would be superficially measurement based. We agree that, while both exploit correlations, MBCC is of an intrinsically different nature to universal MBQC. Thus, it is all the more remarkable that it exhibits non-classical computational attributes due to the correlations. 

Can we identify further instances of MBCC whose efficient classical simulation would lead to hierarchy collapse, or which provide other evidence of non-classical computation? Rieffel and Wiseman \cite{howard} suggest a comparison with the output distribution of Shor's order finding algorithm. It is very likely that the algorithm cannot be efficiently simulated by a classical computing device. There are, however, a number of key differences between such a distribution and those considered here. In particular, the output distribution of the order finding algorithm is of a simple form, sharply peaked at integer multiples of the inverse order $r^{-1}$ \cite{Shor}. A single distribution, therefore, only enables one instance of order-finding, and is not a generic resource for order-finding problems. Further, given order $r$ (extractable after a few samples of the distribution) one can now efficiently (albeit non-exactly) sample from the distribution using a standard classical computer -- knowledge of $r$ has removed the  hardness of the problem. Strikingly, post-selection itself can then also be efficiently classically simulated, since, given $r$,  all non-zero outputs of the distribution can be efficiently classically computed and uniformly sampled from. The need to efficiently extract useful information from the output distributions of standard quantum algorithms is in tension with the properties needed to enable powerful computations under post-selection.% possessed by \IQPstar distributions. %It remains remarkable that, from a sampling perspective, the simple circuits which generate \IQPstar distributions are 

%Finally, a single instance of Shor's algorithm has a fixed input whereas a single distribution allows for \IQPstar computations via MBCC on any input bitstring (of a particular length). We emphasise that the dependence on input for \IQPstar is simple (due to equation \ref{mainequation}) it is far from trivial under post-selection, enabling highly non-trivial computations on that input.

\textit{Conclusions.} 
Theorem 1 provides strong evidence that MBCC cannot be simulated efficiently on a gate-model classical computer whose sole randomness source is a supply of uniformly random bits. The MBCC model can exploit the correlations in probability distributions to achieve non-classical computation. This seems paradoxical, since MBCC is described in fully classical terms. 
The resolution of this apparent paradox is that this probability distribution can be created by quantum means. From a computational perspective, any family of distributions which do not have an efficient classical implementation have a non-classical character. For a probability distribution,  the key marker of non-classicality is typically the violation of a Bell inequality. MBCC reminds us that distributions which violate no Bell inequality may still possess a non-classical character.

From the perspective of quantum information, these results have direct implications for MBQC. They provide a concise argument that adaptive measurement is necessary for universal MBQC (otherwise MBQC and MBCC would be equivalent), and demonstrate that, in spite of this, adaptive measurement is not required for the model to exhibit characteristics of non-classical computation.

Does this work have experimental relevance? Can one \textit{verify} that the desired probability distribution has been produced? This question has been considered for boson sampling \cite{jenssample,aa} and remains a topic for future work. A small-scale demonstration experiment of the key principles of Theorem 1 would be to implement post-selection via experimental repetition and confirm a non-trivial post-selected computation.  Another important issue associated with experimental realizations is that of approximation. While the approximation model (multiplicative errors on individual probabilities) considered in Ref. \cite{IQP2} can be immediately applied here, it does not provide an analysis of the most physically relevant error model, namely, additive error on the whole probability distribution. We note that AA \cite{aa} have made progress in this direction.

The relationship between entanglement and quantum computational speed-up has been debated since the early days of quantum computing theory. These results emphasize the intricacies of this question. They illustrate that the quantum speed up in Shor's algorithm is of a qualitatively different character to the hardness of the sampling distributions in this paper, and show that even an ostensibly fully classical model, such as MBCC, can have a quantum computational character if its resource distribution is quantum mechanically generated. The relationship between quantum computational power and correlations appears even more subtle than previously thought.

\textit{Acknowledgements}---We thank Jonathan Oppenheim, Richard Jozsa, Fernando Brand\~{a}o for insightful comments, Eleanor Rieffel and Howard Wiseman for stimulating correspondence and Jens Eisert and Michael Bremner for helpful discussions. MJH acknowledges funding from CHIST ERA (DIQIP), JJW was supported by the IARPA MQCO program and by the ARC via EQuS project number CE11001013, HA and NU are funded by the EPSRC, RR acknowledges support from NSERC, MITACS and Cifar and DEB is supported by the Leverhulme Trust.

\newpage
\clearpage
\begin{appendix}

\section{Some computational complexity definitions}\label{appsec0}

It is important that we define what we mean by ``classical computing devices". By this expression we include general probabilistic processes since quantum circuits will, in general, not give deterministic outcomes and we want to simulate the statistics of quantum circuits. We follow \cite{IQP2} in making the following definition (where Boolean circuits are defined in \cite{papabook}):

\begin{definition}
A \textbf{standard classical computing device} takes an input $x\in\{0,1\}^{n}$ of size $n$ and produces a bit string $y\in\{0,1\}^{s}$ of size $s$ by performing the following operations:
\begin{enumerate}
\item Flip $r_{1}$ fair coins to produce a random bit string $z\in\{0,1\}^{r_{1}}$ of size $r_{1}$ where we make the assignment ``heads" to $0$ and ``tails" to $1$;
\item Prepare the bit string $x'=(x,\tilde{0},z)$ where $\tilde{0}$ is the bit string of size $r_{2}$ with all elements equal to $0$; and
\item Apply a Boolean circuit $B_{n}$ that takes bit string $x'$ of size $n+r_{1}+r_{2}$ as an input and outputs the bit string $y\in\{0,1\}^{s}$ of size $s$, where the description of $B_{n}$ is generated in poly($n$) by a classical Turing Machine.
\end{enumerate}
All variables $n$, $s$, $r_{1}$, $r_{2}$ are positive integers. If the computational resources of this device are polynomial in $n$ then i.e. $s=\text{poly}(n)$, $r_{1}=\text{poly}(n)$, and $r_{2}=\text{poly}(n)$.
\end{definition}

There is a classical computational complexity class associated with these devices if the computational resources they use are polynomial in $n$, and this is the class of decision problems known as BPP. Since these classical computing devices involve nondeterministic processes, the correct answer to a decision problem may not be obtained deterministically; there is a probability of giving the wrong answer. A problem is in BPP if this error probability is less than or equal to some constant $c<1/2$. A related and potentially larger complexity class of probabilistic classical computations is PP where the error probability is bounded from above by $1/2$ but may not be a constant, i.e. can depend on the input size $n$. This last class will also be important in later discussion.

Similar to BJS \cite{IQP2} and AA \cite{aa}, our proof makes use of complexity classes defined using \emph{post-selection}. Post-selection cannot be achieved deterministically in a physical realization, but is an extremely useful technical concept. Post-selection is the act of demanding that the outcome of a quantum measurement is a fixed value, and that the state of the system then evolves via a fixed projector and a renormalization. Aaronson introduced the class $\textsc{PostBQP}$, which, loosely speaking, represents the problems solvable on a quantum computer given polynomial resources if we were given the extra power of post-selecting measurement outcomes. BJS introduced the complexity class $\textsc{PostIQP}$, which applies a similar treatment to their uniform IQP circuits.

The formal definitions of the classes $\textsc{PostBQP}$, $\textsc{PostIQP}$ and $\textsc{PostIQP}^{*}$ are as follows.

\begin{definition}
A \textbf{BQP circuit family} $\{C_{n}:n\in\mathbb{N}\}$ is a set of quantum circuits such that for each input bit string $x$ of size $n\in\mathbb{N}$, $C_{n}$ is a quantum circuit acting on $q=\text{poly}(n)$ qubits (initiated in the state $|x\rangle|0\rangle^{\otimes{q-n}}$) with a sequence of gates chosen from the universal gate set $\{CZ,H,Z,P\}$ which has a description generated in $\text{poly}(n)$ time by a classical Turing Machine.
\end{definition}

Here $CZ$ is the controlled-$Z$ gate, $H$ is the Hadamard gate, $Z$ is the Pauli-$Z$ gate and $P=e^{i\frac{\pi}{8}Z}$. We now define the aforementioned complexity classes under post-selection utilizing the definitions above and following the work of Aaronson in \cite{aaronson3}.

\begin{definition}\label{postbqp} A language $L\subseteq\{0,1\}^{*}$ is in $\textsc{PostBQP}$ iff there exists a BQP circuit family $\{C_{n}|n\in\mathbb{N}\}$ such that for all inputs $x\in\{0,1\}^{n}$:
\begin{enumerate}
\item after $C_{n}$ is applied to the state $|x\rangle|0\rangle|0\rangle...|0\rangle$, there is a non-zero probability that $a$ qubits (excluding the last qubit) at the end of the circuit are in the state $|0\rangle^{\otimes a}$ for $a=poly(n)$;
\item if $x\in L$ and these $a$ qubits are in the state $|0\rangle^{\otimes a}$, the last qubit when measured in the computational basis is $|1\rangle$ with probability $\geq2/3$;
\item if $x\notin L$ and these $a$ qubits are in the state $|0\rangle^{\otimes a}$, the last qubit when measured in the computational basis is $|1\rangle$ with probability $\leq1/3$;
\end{enumerate}
\end{definition}
\noindent where the last qubit is the bottom qubit line in a quantum circuit diagram such as in Fig. \ref{fig:egIQP}. These are postselected circuits because we postselect on a number of qubits all being in the state $|0\rangle$, and then accept the outcome of a measurement on the last qubit. Therefore, the set of qubits upon which we postselect never includes the last qubit.

To define $\textsc{PostIQP}$, BJS use a near identical definition, replacing ``a BQP circuit family" with their definition of ``a uniform IQP circuit family" in Definition \ref{postbqp}. Similarly, we define $\textsc{Post\IQPstar}$ by replacing, ``a BQP circuit family" in this definition by ``an {\IQPstar} circuit family".

\section{Proof of Lemma \ref{lemma1}}\label{appsec1}

In this section we provide a proof of Lemma \ref{lemma1}. This section utilizes some of the technical concepts defined in Appendix \ref{appsec0}. To prove Lemma \ref{lemma1} we use the arguments {presented} in \cite{IQP2}. 

One of the key lemmas which underpin BJS's main result (see Corollary 3.3 in \cite{IQP2}) is the complexity class equation:
\begin{lemma}
\textsc{PostIQP} = \textsc{PostBQP} = \textsc{PP}
\end{lemma}
\noindent where $\textsc{PostIQP}$, $\textsc{PostBQP}$ and PP are defined in {Appendix} \ref{appsec0} and {in} \cite{IQP2}. The right-hand equation is due to Aaronson \cite{aaronson3}.
 
To prove {Lemma} \ref{lemma1}, we need to show that{:} 
\begin{lemma}\label{newlemma}
\textsc{Post\IQPstar}=\textsc{PostBQP}=\textsc{PP}.
\end{lemma}
Lemma 1 then follows directly from all other steps of the proof of {Corollary} 3.3 in \cite{IQP2}. We shall not reproduce those steps of the proof here.

{In order to prove $\textsc{Post\IQPstar}=\textsc{PostBQP}$, it is necessary to prove that $\textsc{Post\IQPstar}\subseteq\textsc{PostBQP}$ and that $\textsc{PostBQP}\subseteq\textsc{Post\IQPstar}$. It is clear that the former is true. To prove the latter, recall that any BQP circuit can be expressed in the formalism of MBQC \cite{RBB}.} The physical realization of any BQP circuit family {in {MBQC}} comprises the generation of a graph state of sufficient size and appropriate structure, {followed by} adaptive single-qubit measurements in the bases $X$, $Y$ and  $(X\pm Y)/\sqrt{2}$.

Graph state generation comprises{:}
\begin{enumerate}
\item preparation of a set of qubits in state $\ket{+}=(\ket{0}+\ket{1})/\sqrt{2}${; and}
\item implementation of CZ gates between certain {qubits.}
\end{enumerate}

The measurements can be implemented by rotating about the $Z$ axis by $0$, $\pi/4$, $-\pi/4$ or $\pi/2$ followed by a measurement in the $X$ basis.

The sequence of measurements corresponds to the chosen BQP circuit. Notice that from Definition 6  the initial state and the gates depend on $n$ and not the specific value of $x$. The dependence on $x$ is introduced by having measurements of the following form (via a simple application of gate identities in \cite{RBB}): there are $n$ qubits that each correspond to each element $x_{j}$ of $x$, and if $x_{j}=0$ we make the measurement in the $X$ basis, and if $x_{j}=1$ then we implement a $\pi$-rotation about the $Z$ axis and measure in the $X$ basis. Since the rotation commutes with the CZ gates, it can be incorporated into the preparation of that qubit by preparing $\ket{+}$ if $x_{j}=0$ and $\ket{-}=\frac{1}{\sqrt{2}}\left(|0\rangle-|1\rangle\right)$ otherwise. For the rest of the qubits in the graph state, the rotation about the $Z$ axis and measurement in the $X$ basis is only defined by size $n$, the size of $x$.

If we could post-select on measurement outcomes in MBQC then we remove the need for adaptive measurements---computations are accepted only if certain outcomes occur. We now show that non-adaptive computations in MBQC are instances of {\IQPstar} circuits. This can be easily shown if one writes out the computation in MBQC as a quantum circuit, and for all qubits after preparation of states $\ket{+}$ and $\ket{-}$ and before measurements in the $X$ basis two Hadamard gates are applied, yielding exactly the same circuit. The action of a Hadamard on these state preparations gives the mapping $\frac{1}{\sqrt{2}}(|0\rangle+(-1)^{x_{j}}|1\rangle)\rightarrow|x_{j}\rangle$. In a computation in MBQC all unitaries prior to measurement are diagonal in the $Z$ basis, so the action of a Hadamard on every qubit prior to and after these unitaries maps these unitaries diagonal in the $Z$ basis to unitaries that are diagonal in the $X$ basis; those unitaries that appear in {\IQPstar} circuits. Finally the action of a Hadamard prior to measurement in the $X$ basis results in a measurement in the computational basis. In Fig. \ref{fig:dualcircuits} we give an example of this equivalence between MBQC circuits without adaptivity and IQP circuits.

\begin{figure}[h!]
\centering
	\includegraphics[width=0.3\textwidth]{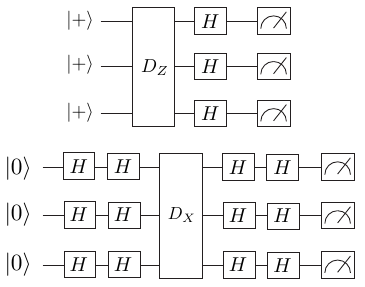}
\caption{The top circuit where $D_{Z}$ is a unitary that is diagonal in the Pauli-Z basis is equal to the bottom circuit where $D_{X}$ is a unitary that is diagonal in the Pauli-X basis. The top circuit is of the form of a circuit in MBQC but without adaptivity. Since $H\cdot H=\openone$ we immediately see that the bottom circuit is an example of an IQP circuit.}
\label{fig:dualcircuits}
\end{figure}

Therefore, under the action of post-selection MBQC computations are instances of {\IQPstar} circuits. To complete the proof of Lemma \ref{lemma1} we insert Lemma \ref{newlemma} into the proofs of Theorem 3.2 in \cite{IQP2} and Corollary 3.3 in \cite{IQP2}.

\end{appendix}


\begin{thebibliography}{99}
\bibitem{Shor} P. Shor, \textit{SIAM Rev.} \textbf{41}, 303 (1999).


\bibitem{vandennest1} M. Van den Nest, Quant. Inf. Comp. \textbf{11}, 784 (2011).

\bibitem{jozsa} R. Jozsa and N. Linden, \textit{Proc. R. Soc. A} \textbf{459}, 2011 (2003).
\bibitem{aa} S. Aaronson and A. Arkhipov, \textit{Proc. of ACM STOC}, 333-342, (2011).
\bibitem{terhaldivincenzo} B. M. Terhal and D. P. DiVincenzo, \textit{Quantum Information and Computation} \textbf{4}, 2, 134 (2004).
\bibitem{IQP2} M. J. Bremner, R. Jozsa, and D. J. Shepherd, \textit{Proc. R. Soc. A} \textbf{467}, 459 (2011).
\bibitem{mbqc} R. Raussendorf and H. J. Briegel, \textit{Phys. Rev. Lett.} \textbf{86}, 5188 (2001).
\bibitem{briegel1} H.J. Briegel and R. Raussendorf, Phys. Rev. Lett. \textbf{86}, 910 (2001).
\bibitem{mbqcresources} M. Van Den Nest, \textit{et al}, Phys. Rev. Lett. 97, 150504 (2006).
\bibitem{andersbrowne} J. Anders and D. E. Browne, \textit{Phys. Rev. Lett.} \textbf{102}, 050502 (2009). 
\bibitem{raussenGHZ} R. Raussendorf, Phys. Rev. A 88, 022322 (2013).
\bibitem{RBB} R. Raussendorf, D. E. Browne and H. J. Briegel, \textit{Phys. Rev. A} \textbf{68}, 022312 (2003).
\bibitem{matty} M. J. Hoban and D. E. Browne, \textit{Phys. Rev. Lett.} \textbf{107}, 120402 (2011).
\bibitem{papabook} C.H. Papadimitriou, \textit{Computational Complexity}, Addison Wesley, Boston (1993).
\bibitem{jozsavdn} R. Jozsa and M. Van Den Nest, arXiv:1305.6190v1.
\bibitem{IQP1} D. J. Shepherd and M. J. Bremner, \textit{Proc. R. Soc. A} \textbf{465}, 1413 (2009).
%\bibitem{supmat} See Supplemental Material at [URL will be inserted by publisher] for a review of key computational complexity theoretic concepts and the full proof of Lemma 1.
\bibitem{aaronson3} S. Aaronson, \textit{Proc. R. Soc. A} \textbf{461}, 3473 (2005).
\bibitem{gbriegel} D. E. Browne and H. J. Briegel, \textit{Lectures on Quantum Information, D. Bru\ss, G. Leuchs (Ed.)}, Wiley-VCH, Berlin, 359 (2006).
\bibitem{cablermp} K. Modi, A. Brodutch, H. Cable, T. Paterek, and V. Vedral, \textit{Rev. Mod. Phys.} \textbf{84}, 1655 (2012).
\bibitem{howard} E. G. Rieffel, H. M. Wiseman, \textit{Phys. Rev. A} \textbf{89}, 032323 (2014).
\bibitem{jenssample} C. Gogolin, M. Kliesch, L. Aolita and J. Eisert, arXiv:1306.3995 (2013).
%\bibitem{bosonsamplingexp} M. A. Broome et al, \textit{Science} \textbf{339}, 794 (2013); M. Tillmann et al, \textit{Nature Photonics} \textbf{7}, 540 (2013); J. B. Spring et al, \textit{Science} \textbf{339}, 798 (2013); A. Crespi et al, \textit{Nature Photonics} \textbf{7}, 545 (2013).
%\bibitem{nielsenchuang} M. A. Nielsen and I. Chuang, \textit{Quantum Computation and Quantum Information}, Cambridge University Press (2000).


%\bibitem{knuth} D. Knuth and A. Yao, \textit{Algorithms and Complexity: New Directions and Recent Results}, Academic Press, (1976).
%\bibitem{eastindiscord} B. Eastin, \textit{arXiv:1006.4402} (2010).

%\bibitem{aaronsongottesman} S. Aaronson, D. Gottesman, Phys. Rev. A \textbf{70}, 052328 (2004).
%\bibitem{onecleanqubit} E. Knill and R. Laflamme, \textit{Phys. Rev. Lett.} \textbf{81}, 5672 (1998).
%\bibitem{tooentangled} D. Gross, S. T. Flammia, and J. Eisert, \textit{Phys. Rev. Lett.} \textbf{102}, 190501 (2009); M. J. Bremner, C. Mora, and A. Winter, \textit{Phys. Rev. Lett.} \textbf{102}, 190502 (2009).
%\bibitem{vandennestugc} M. Van den Nest, A. Miyake, W. D\"{u}r, and H. J. Briegel, \textit{Phys. Rev. Lett}. 97, 150504 (2006).
%\bibitem{nmbqc} M. J. Hoban et al, \textit{New J. Physics} \textbf{13}, 023014 (2011).
%\bibitem{mattyjoel} M. J. Hoban, J. J. Wallman and D. E. Browne, \textit{Phys. Rev. A} \textbf{84}, 062107 (2011). 
%\bibitem{perdrix} D. E. Browne, E. Kashefi, and S. Perdrix, %\textit{Proc. of Theory of Quantum Computation, Communication, and Cryptography (TQC 2010)}, 35-46 (2011).
%\bibitem{bell} J. S. Bell, \textit{Speakable and Unspeakable in Quantum Mechanics}, Cambridge University Press, Cambridge (1988).
%\bibitem{shepherd} D. J. Shepherd, \textit{arXiv:1005.1744} (2010).
%\bibitem{ghz} D. M. Greenberger, M. A. Horne, A. Zeilinger, \textit{Bell's Theorem, Quantum Theory, and Conceptions of the Universe}, M. Kafatos (Ed.), Kluwer, Dordrecht, 69-72 (1989).
\end{thebibliography}
\end{document}